\documentclass[english,5p, preprint]{elsarticle}
\usepackage[T1]{fontenc}
\usepackage[latin9]{inputenc}
\usepackage{amsmath}
\usepackage{amssymb}
\usepackage{graphicx}

\makeatletter
\usepackage{times}   
\newtheorem{theorem}{Theorem}[section]
\newtheorem{remark}[theorem]{Remark}

\usepackage[all]{xy}

\usepackage{balance}
\usepackage{psfrag}
\usepackage{dsfont}

\newenvironment{proof}{\textbf{Proof. }}{\hfill{$\square$}}

\usepackage{color}

\makeatother

\usepackage{babel}
\begin{document}
\begin{frontmatter}
\title{Stabilizing Unstable Periodic Orbits with Delayed Feedback Control \\ in Act-and-Wait Fashion} 
\author[First]{Ahmet Cetinkaya}
\author[Second]{Tomohisa Hayakawa}
\author[Second]{Mohd Amir Fikri bin Mohd Taib}

\address[First]{Department of Computer Science,  \\ \vskip 1pt Tokyo Institute of Technology, Yokohama  226-8502, Japan} 
\address[Second]{Department of Systems and Control Engineering, \\ \vskip 1pt Tokyo Institute of Technology, Tokyo 152-8552, Japan \\ \vskip 3pt \textup{\texttt{ahmet@sc.dis.titech.ac.jp, hayakawa@sc.e.titech.ac.jp, amir@dsl.mei.titech.ac.jp}}}

\begin{abstract}
A delayed feedback control framework for stabilizing unstable periodic orbits of linear periodic time-varying systems is proposed. In this framework, act-and-wait approach is utilized for switching a delayed feedback controller on and off alternately at every integer multiples of the period of the system. By analyzing the monodromy matrix of the closed-loop system, we obtain conditions under which   the closed-loop system's state converges towards a periodic solution under our proposed control law. We discuss the application of our results in stabilization of unstable periodic orbits of nonlinear systems and present numerical examples to illustrate the efficacy of our approach. 
\end{abstract}


\end{frontmatter} 

\section{Introduction}

Stabilization of unstable periodic orbits of nonlinear systems using
delayed feedback control was first explored in \cite{pyragas:ccc}.
In the delayed feedback control scheme, the difference between the
current state and the delayed state is utilized as a control input
to stabilize an unstable orbit. The delay time is set to correspond
to the period of the orbit to be stabilized so that the control input
vanishes when the stabilization is achieved. 

Delayed feedback controllers have been used in many studies for stabilization
of the periodic orbits of both continuous- and discrete-time nonlinear
systems (see, e.g., \cite{harrington:drd,tian2005survey,hovel2010control},
and the references therein). More recently, \cite{pyragas2013adaptive}
investigated delayed feedback control of nonlinear systems that are
subject to noise, \cite{fiedler2016delayed} explored delayed feedback
control of a delay differential equation, and \cite{ichinose2014delayed}
utilized delayed feedback control for stabilizing quasi periodic orbits.
The work \cite{pyragas2015relation} studied the relation between
the delayed feedback control approach and the harmonic oscillator-based
control methods for stabilizing periodic orbits in chaotic systems
\cite{olyaei2015controlling}. Furthermore, \cite{novivcenko2012phase}
and \cite{purewal2014effect} explored the situation where the period
of the orbit and the delay time in the delayed feedback controller
do not match due to imperfect information about the periodic orbit
or inaccuracies in the implementation of the controller. 

The physical structure of delayed feedback control scheme is simple.
However, the analysis of the closed-loop system is difficult. This
is due to the fact that to investigate the system under delayed feedback
control, one has to deal with delay-differential equations, the state
space of which is infinite dimensional. To deal with the difficulties
in the analysis of delay differential equations, an approach is to
use approximation techniques (see, for instance, \cite{ma2003} and
\cite{butcher2004}). Another approach was taken in \cite{insperger:aaw}.
There, stabilization of a linear time-invariant system with a time-delay
controller was considered, and ``act-and-wait'' concept was introduced.
This concept is characterized by alternately applying and cutting
off the controller in finite intervals. It is shown in \cite{insperger:aaw}
that by utilizing the act-and-wait concept, one may be able to derive
a finite-sized monodromy matrix for the closed-loop system, which
can then be used for stability analysis. Act-and-wait concept has
been extended to discrete-time systems in \cite{insperger:acd}, and
tested through experiments in \cite{insperger:iad}. Furthermore,
act-and-wait approach has been used together with delayed feedback
control in \cite{konishi2011delayed} for stabilizing unstable fixed
points of nonlinear systems, and more recently in \cite{pyragas2016}
for stabilizing unstable periodic orbits of nonautonomous nonlinear
systems. 

In this paper, we explore the stabilization of periodic solutions
to linear periodic systems with an act-and-wait-fashioned delayed
feedback control framework. ~In this framework, a switching mechanism
is utilized to turn the delayed feedback controller on and off alternately
at every integer multiple of the period of a given linear periodic
system. Act-and-wait scheme allows us to obtain the monodromy matrix
associated with the closed-loop system under our proposed controller.
We then use the obtained monodromy matrix for obtaining conditions
under which the closed-loop system's state converges to a periodic
solution. Our main motivation for studying a delayed feedback control
problem for periodic systems stems from our desire to analyze the
stability of a periodic orbit of a nonlinear system under delayed
feedback control. In this paper we apply our results for linear periodic
systems in analyzing periodic linear variational equations obtained
after linearizing nonlinear systems (under delayed feedback control)
around periodic trajectories corresponding to periodic orbits. The
uncontrolled nonlinear systems that we consider are autonomous and
as a result their stability assessment under the act-and-wait-fashioned
delayed feedback controller differs from the nonautonomous case discussed
in \cite{pyragas2016}. We also note that our delayed feedback control
approach and therefore our analysis techniques differ from those in
earlier works on stabilization of linear periodic systems where researchers
have employed Gramian-based controllers \cite{montagnier2004}, periodic
Lyapunov functions \cite{zhou2012}, and linear matrix inequalities
\cite{de2000lmi}.

The paper is organized as follows. In Section~\ref{sec:mainresults},
we introduce our act-and-wait-fashioned delayed feedback control framework
for stabilizing periodic solutions of linear periodic systems; we
present a method for assessing the asymptotic stability of a periodic
solution of the closed-loop system under our proposed framework. Furthermore,
in Section~\ref{sec:Stabilization-of-Unstable} we discuss an application
of our results in stabilizing unstable periodic orbits of nonlinear
systems. We present illustrative numerical examples in Section~\ref{sec:Illustrative-Numerical-Example}.
Finally, we conclude our paper in Section~\ref{sec:Conclusion}. 

We note that a preliminary version of this work was presented in \cite{amirhayakawacetinkaya2013}.
In this paper, we provide additional discussions and examples. 

\section{Delayed Feedback Stabilization of Periodic Orbits}

\label{sec:mainresults}

In this section, we provide the mathematical model for a linear periodic
time-varying system and introduce a new delayed feedback control framework
based on act-and-wait approach. We then characterize a method for
evaluating convergence of state trajectories of a closed-loop linear
time-varying periodic system towards a periodic solution. 

\subsection{Linear Periodic Time-Varying System \label{subsec:Linear-Periodic-Time-Varying}}

Consider the linear periodic time-varying system 
\begin{align}
\dot{x}(t) & =A(t)x(t)+B(t)u(t),\quad x(t_{0})=x_{0},\quad t\geq t_{0},\label{eq:system}
\end{align}
where $x(t)\in\mathbb{R}^{n}$ is the state vector, $u(t)\in\mathbb{R}^{m}$
is the control input, and $A(t)\in\mathbb{R}^{n\times n}$ and $B(t)\in\mathbb{R}^{n\times m}$
are periodic matrices with period $T>0$, that is, $A(t+T)=A(t)$
and $B(t+T)=B(t)$, $t\geq t_{0}$. For simplicity of exposition,
we assume $t_{0}=0$ for the rest of the discussion because the case
where $t_{0}\neq0$ can be similarly handled. Furthermore, we assume
that the uncontrolled ($u(t)\equiv0$) dynamics possess a periodic
solution $x(t)\equiv x^{*}(t)$ with period $T$ satisfying $x^{*}(t+T)=x^{*}(t),\,t\geq0$.
It follows from Floquet's theorem that there exists such a periodic
solution to the uncontrolled system \eqref{eq:system} of period $T$
if and only if there exists a nonsingular matrix $C\in\mathbb{R}^{n\times n}$
possessing $1$ in its spectrum such that $V(t+T)=V(t)C$, where $V(t)$
denotes a fundamental matrix of the uncontrolled system \eqref{eq:system}.
Moreover, note that since $x(t)\equiv x^{*}(t)$ is a $T$-periodic
solution of the uncontrolled system \eqref{eq:system}, $x(t)\equiv\alpha x^{*}(t)$
is also a $T$-periodic solution for all $\alpha\in\mathbb{R}$, that
is, $x(t)\equiv\alpha x^{*}(t)$ satisfies \eqref{eq:system} with
$u(t)\equiv0$. 

We investigate the asymptotic stability of periodic solutions of the
closed-loop system \eqref{eq:system} under the delayed feedback control
input 
\begin{align}
u(t) & =-g(t)F(x(t)-x(t-T)),\label{eq:control-input}
\end{align}
where $F\in\mathbb{R}^{m\times n}$ is a constant gain matrix and
\begin{align}
g(t) & \triangleq\left\{ \begin{array}{rl}
0, & 2kT\leq t<(2k+1)T,\\
1, & (2k+1)T\leq t<2(k+1)T,
\end{array}\right.\,k\in\mathbb{N}_{0}.\label{eq:gdef}
\end{align}
is a time-varying function that switches the controller on and off
alternately at every integer multiples of the period $T$. 

Note that the feedback term characterized in \eqref{eq:control-input}
vanishes after the periodic solution is stabilized. Specifically,
for $x(t)\equiv x^{*}(t)$, we have $u(t)=0$, $t\geq0$, since $x(t)=x(t-T)$. 

We remark that our control approach is a specific case of the act-and-wait
approach introduced in \cite{insperger:aaw}. In particular, in our
control law \eqref{eq:control-input}, both the acting and the waiting
durations have length $T$. Specifically, in every $2T$ period, the
controller first waits for a duration of length $T$, and then acts
for a duration of length $T$. Note that the controllers in \cite{insperger:aaw}
are more general in the sense that acting and waiting times need not
be equal. In Section~\ref{sec:Illustrative-Numerical-Example}, we
also consider different switching functions $g(t)$ that lead to different
acting and waiting times. 

The reason why $g(t)$ is set to be a time-varying function can be
understood if we compare it to the case where $g(t)$ is constant.
For instance, if $g(t)\equiv1$ in \eqref{eq:control-input}, then
\eqref{eq:system} becomes 
\begin{align}
\dot{x}(t) & =(A(t)-B(t)F)x(t)+B(t)Fx(t-T),\label{eq:dde}
\end{align}
which is a delay-differential equation. Analysis of the solution of
\eqref{eq:dde} is difficult, as the state space associated with \eqref{eq:dde}
is infinite-dimensional.

On the other hand, for the linear periodic system 
\begin{align}
\dot{x}(t) & =A(t)x(t),\quad A(t)=A(t+T),\label{eq:controlfreesystem}
\end{align}
where there are no delay terms, stability of an equilibrium solution
can be assessed by analyzing the corresponding monodromy matrix. Let
$\Phi(\cdot,\cdot)$ denote the state-transition matrix of \eqref{eq:controlfreesystem}.
The monodromy matrix associated with the $T$-periodic system \eqref{eq:controlfreesystem}
is given by $\Phi(T,0)\in\mathbb{R}^{n\times n}$. The eigenvalues
of the monodromy matrix, known as the Floquet multipliers, are essential
in the analysis of the long-term behavior of the state-transition
matrix of \eqref{eq:controlfreesystem}, because 
\begin{align}
\Phi(t+kT,0) & =\Phi(t,0)\Phi^{k}(T,0),\quad k\in\mathbb{N}_{0}.\label{eq:phiaproperty}
\end{align}
Moreover, the state of the periodic system \eqref{eq:controlfreesystem}
satisfies 
\begin{align*}
x((k+1)T) & =\Phi x(kT),\quad k\in\mathbb{N}_{0}.
\end{align*}
 Observe that if $g(t)\equiv1$ in \eqref{eq:control-input}, we would
not be able to find a homogeneous expression in the form of \eqref{eq:controlfreesystem},
let alone find a corresponding ``monodromy matrix'', because of
the existence of the delay term. 

However, in our case, following the act-and-wait approach, we define
$g(t)$ as in \eqref{eq:gdef} as a switching function. Consequently,
we are able to construct a monodromy matrix $\Lambda\in\mathbb{R}^{n\times n}$
for the closed-loop system \eqref{eq:system}, \eqref{eq:control-input}
with the \emph{doubled period} $2T$ such that 
\begin{align}
x(2(k+1)T) & =\Lambda x(2kT),\quad k\in\mathbb{N}_{0}.\label{eq:monodromy-equation}
\end{align}
Note that the spectrum of the monodromy matrix $\Lambda$ characterizes
long-term behavior of the state trajectory. 

In the following sections, we first derive the monodromy matrix, and
then we present conditions for the convergence of the state trajectory
towards a periodic solution of the  closed-loop system \eqref{eq:system},
\eqref{eq:control-input}. 

\subsection{Monodromy Matrix\label{subsec:Monodromy-Matrix} }

In this section, we obtain the monodromy matrix associated with the
closed-loop system given by \eqref{eq:system}, \eqref{eq:control-input}.
In our derivations, we use $\Phi(\cdot,\cdot)$ to denote the state-transition
matrix associated with \eqref{eq:controlfreesystem}. Furthermore,
let $\Upsilon(\cdot,\cdot)$ denote the state-transition matrix for
the linear $T$-periodic system 
\begin{align}
\dot{x}(t) & =(A(t)-B(t)F)x(t).\label{eq:delayfreecontrolsystem}
\end{align}

Now, let $\mathcal{T}_{0}(k)\triangleq[2kT,(2k+1)T)$, $\mathcal{T}_{1}(k)\triangleq[(2k+1)T,2(k+1)T)$,
$k\in\mathbb{N}_{0}$. Note that when $t\in\mathcal{T}_{1}(k)$, the
controller is on, that is, $g(t)=1$. Hence, it follows from \eqref{eq:system}
and \eqref{eq:control-input} that for $t\in\mathcal{T}_{1}(k)$,
\begin{align}
\dot{x}(t) & =(A(t)-B(t)F)x(t)+B(t)Fx(t-T).\label{eq:usegeq1}
\end{align}
Observe that for $t\in\mathcal{T}_{1}(k)$, we have $t-T\in\mathcal{T}_{0}(k)$.
Since the controller is turned off during the interval $\mathcal{T}_{0}(k)$,
the evolution of the state in this interval is described by \eqref{eq:controlfreesystem}
corresponding to the uncontrolled dynamics. Therefore, $x(t-T)$ can
be expressed by 
\begin{align}
x(t-T) & =\Phi(t-T,2kT)x(2kT),\quad t\in\mathcal{T}_{1}(k).\label{eq:xtminusT}
\end{align}
Now, by using \eqref{eq:usegeq1} and \eqref{eq:xtminusT}, we obtain
\begin{align}
\dot{x}(t) & =(A(t)-B(t)F)x(t)\nonumber \\
 & \quad+B(t)F\Phi(t-T,2kT)x(2kT),\quad t\in\mathcal{T}_{1}(k).\label{eq:xdiffwithbtf}
\end{align}
By multiplying both sides of \eqref{eq:xdiffwithbtf} from left with
the matrix $\Upsilon^{-1}(t,(2k+1)T)$, we obtain 
\begin{align}
 & \Upsilon^{-1}(t,(2k+1)T)\dot{x}(t)\nonumber \\
 & \,\,=\Upsilon^{-1}(t,(2k+1)T)(A(t)-B(t)F)x(t)\nonumber \\
 & \,\,\quad+\Upsilon^{-1}(t,(2k+1)T)B(t)F\Phi(t-T,2kT)x(2kT).\label{eq:upsiloninverse}
\end{align}
Since $\Upsilon^{-1}(t,(2k+1)T)(A(t)-B(t)F)=-\frac{\mathrm{d}}{\mathrm{d}t}\Upsilon^{-1}(t,(2k+1)T)$,
we have 
\begin{align}
 & \Upsilon^{-1}(t,(2k+1)T)\dot{x}(t)\nonumber \\
 & \quad-\Upsilon^{-1}(t,(2k+1)T)(A(t)-B(t)F)x(t)\nonumber \\
 & =\Upsilon^{-1}(t,(2k+1)T)\dot{x}(t)+\frac{\mathrm{d}}{\mathrm{d}t}\Upsilon^{-1}(t,(2k+1)T)x(t)\nonumber \\
 & =\frac{\mathrm{d}}{\mathrm{d}t}\big(\Upsilon^{-1}(t,(2k+1)T)x(t)\big).\label{eq:derivationparts}
\end{align}
It follows from \eqref{eq:upsiloninverse} and \eqref{eq:derivationparts}
that 
\begin{align}
 & \frac{\mathrm{d}}{\mathrm{d}t}\big(\Upsilon^{-1}(t,(2k+1)T)x(t)\big)\nonumber \\
 & \quad=\Upsilon^{-1}(t,(2k+1)T)B(t)F\Phi(t-T,2kT)x(2kT).\label{eq:differential}
\end{align}
Next, we integrate both sides of \eqref{eq:differential} over the
interval $[(2k+1)T,t)$ to obtain 
\begin{align}
 & \Upsilon^{-1}(t,(2k+1)T)x(t)\nonumber \\
 & \,\,=\Upsilon^{-1}((2k+1)T,(2k+1)T)x((2k+1)T)\nonumber \\
 & \,\,\quad+\Big(\int_{(2k+1)T}^{t}\Upsilon^{-1}(s,(2k+1)T)B(s)F\Phi(s-T,2kT)\mathrm{d}s\Big)\nonumber \\
 & \,\,\quad\cdot x(2kT).
\end{align}
Noting that $\Upsilon^{-1}((2k+1)T,(2k+1)T)=I_{n}$ and $x((2k+1)T)=\Phi((2k+1)T,2kT)x(2kT)$,
we obtain 
\begin{align}
 & x(t)=\Upsilon(t,(2k+1)T)\Phi((2k+1)T,2kT)x(2kT)\nonumber \\
 & \quad+\Upsilon(t,(2k+1)T)\nonumber \\
 & \quad\cdot\Big(\int_{(2k+1)T}^{t}\Upsilon^{-1}(s,(2k+1)T)B(s)F\Phi(s-T,2kT)\mathrm{d}s\Big)\nonumber \\
 & \quad\cdot x(2kT).\label{eq:integral-equation}
\end{align}
Since both \eqref{eq:controlfreesystem} and \eqref{eq:delayfreecontrolsystem}
are $T$-periodic, we have $\Upsilon(t,(2k+1)T)=\Upsilon(t-(2k+1)T,0)$
and $\Phi((2k+1)T,2kT)=\Phi(T,0)$. Furthermore, since $\Upsilon(t,(2k+1)T)=\Upsilon(t,s)\Upsilon(s,(2k+1)T)$,
$s\in[(2k+1)T,t]$, we have $\Upsilon(t,(2k+1)T)\Upsilon^{-1}(s,(2k+1)T)=\Upsilon(t,s)$.
Consequently, it follows from \eqref{eq:integral-equation} that 
\begin{align}
 & x(t)=\Big(\Upsilon(t-(2k+1)T,0)\Phi(T,0)\nonumber \\
 & \,\,+\int_{(2k+1)T}^{t}\Upsilon(t,s)B(s)F\Phi(s-T,2kT)\mathrm{d}s\Big)x(2kT),\label{eq:integral-equation-1}
\end{align}
 for $t\in\mathcal{T}_{1}(k)$. We now change the variable of the
integral term in \eqref{eq:integral-equation-1} by setting $\bar{s}=s-2kT$.
As a result, we obtain 
\begin{align}
 & \int_{(2k+1)T}^{t}\Upsilon(t,s)B(s)F\Phi(s-T,2kT)\mathrm{d}s\nonumber \\
 & \quad=\int_{T}^{t-2kT}\Upsilon(t,\bar{s}+2kT)B(\bar{s}+2kT)\nonumber \\
 & \quad\quad\cdot F\Phi(\bar{s}+2kT-T,2kT)\mathrm{d}\bar{s},\quad t\in\mathcal{T}_{1}(k).\label{eq:firstintegralresult}
\end{align}
By $T$-periodicity of \eqref{eq:controlfreesystem} and \eqref{eq:delayfreecontrolsystem},
we have $\Upsilon(t,\bar{s}+2kT)=\Upsilon(t-2kT,\bar{s})$ and $\Phi(\bar{s}+2kT-T,2kT)=\Phi(\bar{s}-T,0)$.
Moreover, $B(\bar{s}+2kT)=B(\bar{s})$, since $B(\cdot)$ is a $T$-periodic
matrix function. It then follows from \eqref{eq:firstintegralresult}
that 
\begin{align}
 & \int_{(2k+1)T}^{t}\Upsilon(t,s)B(s)F\Phi(s-T,2kT)\mathrm{d}s\nonumber \\
 & \quad=\int_{T}^{t-2kT}\Upsilon(t-2kT,\bar{s})B(\bar{s})F\Phi(\bar{s}-T,0)\mathrm{d}\bar{s}.\label{eq:integrationresult}
\end{align}
 Now, from \eqref{eq:integral-equation-1} and \eqref{eq:integrationresult},
we obtain 
\begin{align}
x(t) & =\Big(\Upsilon(t-(2k+1)T,0)\Phi(T,0)\nonumber \\
 & \quad+\int_{T}^{t-2kT}\Upsilon(t-2kT,\bar{s})B(\bar{s})F\Phi(\bar{s}-T,0)\mathrm{d}\bar{s}\Big)\nonumber \\
 & \quad\cdot x(2kT),\quad t\in\mathcal{T}_{1}(k).\label{eq:xtfinalversion}
\end{align}
 By continuity of the state, we can compute $x(2(k+1)T)$ by using
\eqref{eq:xtfinalversion}. Specifically, we set $t=2(k+1)T$ in \eqref{eq:xtfinalversion}
and obtain \eqref{eq:monodromy-equation}, where the monodromy matrix
$\Lambda$ is given by 
\begin{align}
\Lambda & =\Big(\Upsilon(T,0)\Phi(T,0)\nonumber \\
 & \quad+\int_{T}^{2T}\Upsilon(2T,\bar{s})B(\bar{s})F\Phi(\bar{s}-T,0)\mathrm{d}\bar{s}\Big).\label{eq:monodromymatrix}
\end{align}
 Notice that \eqref{eq:monodromy-equation} characterizes the evolution
of the state at times $t=2kT$, $k\in\mathbb{N}_{0}$. Consequently,
stability of the equilibrium solutions of the closed-loop system \eqref{eq:system},
\eqref{eq:control-input} can be deduced through the eigenvalues of
the monodromy matrix. 

Moreover, note that $x^{*}(t)$ satisfies $x^{*}(2T)=x^{*}(0)$. In
addition, from \eqref{eq:monodromy-equation} we have 
\begin{align}
x^{*}(2T) & =\Lambda x^{*}(0).
\end{align}
It follows that $x^{*}(0)=\Lambda x^{*}(0)$, and hence, the monodromy
matrix $\Lambda$ associated with the closed-loop system \eqref{eq:system},
\eqref{eq:control-input} possesses $1$ as an eigenvalue with the
eigenvector $x^{*}(0)$. Note that both the algebraic and the geometric
multiplicity of the eigenvalue $1$ may be greater than $1$. 

Let $\kappa\in\{1,2,\ldots,n\}$ denote the algebraic multiplicity
of the eigenvalue $1$. We represent generalized eigenvectors of the
monodromy matrix $\Lambda$ by vectors $v_{1},v_{2},\ldots,v_{n}\in\mathbb{C}^{n}$
out of which $v_{1},v_{2}\ldots,v_{\kappa}\in\mathbb{R}^{n}$ denote
the generalized eigenvectors associated with the eigenvalue $1$.
The generalized eigenvectors $v_{1},v_{2},\ldots,v_{n}$ are linearly
independent \citep{bernstein2011matrix}, and hence form a basis for
$\mathbb{C}^{n}$, that is, for any $y\in\mathbb{C}^{n}$, there exist
$\alpha_{1}$, $\alpha_{2}$, $\ldots$, $\alpha_{n-1}$, $\alpha_{n}\in\mathbb{C}$
such that $y=\sum_{i=1}^{n}\alpha_{i}v_{i}$. Note that $\alpha_{i}v_{i}$
characterizes the component of $y$ along $v_{i}$ (the $i$th element
of the basis). Note also that linear independence of vectors $v_{1},v_{2},\ldots,v_{n}$
guarantees that the constants $\alpha_{1},\alpha_{2},\ldots,\alpha_{n}$
are uniquely determined. 

The long-term behavior of the state trajectory is determined by the
spectrum of the monodromy matrix $\Lambda$. In Theorem~\ref{theo1}
below, we present the condition for the convergence of the state trajectory
towards a periodic solution. 

\begin{theorem}\label{theo1} Consider the linear time-varying periodic
system \eqref{eq:system} with the periodic solution $x^{*}(t)$.
Let the initial condition be given by $x(0)=x_{0}\triangleq\sum_{i=1}^{n}\alpha_{i}v_{i}$,
where $v_{1},v_{2},\ldots,v_{n}\in\mathbb{C}^{n}$ are the generalized
eigenvectors of the monodromy matrix $\Lambda\in\mathbb{R}^{n\times n}$,
and $\alpha_{1},\alpha_{2},\ldots,\alpha_{n}\in\mathbb{C}$. Suppose
that $1$ is a semisimple eigenvalue of $\Lambda$. Let $\kappa\in\{1,2,\ldots,n\}$
denote the algebraic multiplicity of the eigenvalue $1$ associated
with the eigenspace spanned by the eigenvectors $v_{1},v_{2}\ldots,v_{\kappa}\in\mathbb{R}^{n}$.
If all the eigenvalues, other than the eigenvalue $1$ of the monodromy
matrix $\Lambda$ are strictly inside the unit circle of the complex
plane, then $\lim_{k\to\infty}x(2kT)=\sum_{i=1}^{\kappa}\alpha_{i}v_{i}$.
\end{theorem}

\begin{proof} First, we define $P\triangleq\left[v_{1},v_{2},\ldots,v_{n}\right]$.
Note that since the columns of the matrix $P$ given by $v_{1}$,
$v_{2}$, $\ldots$, $v_{n}$ are linearly independent, it follows
that $P$ is nonsingular. 

Furthermore, note that with the similarity transformation $J\triangleq P^{-1}\Lambda P$,
we obtain the Jordan form\footnote{Here without loss of generality
we are considering the case where in the construction of $P$, the
generalized eigenvectors associated with each eigenvalue are next
to each other and they are in the same order that they appear in the
Jordan chain (see \cite{bernstein2011matrix}) associated with that
eigenvalue. Generalized eigenvectors can always be reordered in this
way to obtain a Jordan form.} of the monodromy matrix $\Lambda$
such that $J=\mathrm{diag}[J_{1},J_{2},\ldots,J_{r}]$, where $r\in\mathbb{N}$
denotes the number of Jordan blocks, which is also equal to the sum
of the geometric multiplicities of the eigenvalues of the monodromy
matrix $\Lambda$. The Jordan blocks $J_{i}\in\mathbb{C}^{n_{i}\times n_{i}}$,
$i\in\{1,2,\ldots,r\}$, have the form $J_{i}=\lambda_{i}I_{n_{i}}+N_{i}$,
where $\lambda_{i}\in\mathbb{C}$ is an eigenvalue of the monodromy
matrix $\Lambda$ and $N_{i}\in\mathbb{R}^{n_{i}\times n_{i}}$ is
a nilpotent matrix of degree $n_{i}$. 

We use \eqref{eq:monodromy-equation}, the definition for the Jordan
form $J$ of the monodromy matrix $\Lambda$, and $[\alpha_{1},\alpha_{2},\ldots,\alpha_{n}]^{\mathrm{T}}=P^{-1}x_{0}$
to obtain 
\begin{align}
x(2kT) & =\Lambda^{k}x_{0}\nonumber \\
 & =PJ^{k}P^{-1}x_{0}\nonumber \\
 & =P\mathrm{diag}[J_{1}^{k},J_{2}^{k},\ldots,J_{r}^{k}][\alpha_{1},\alpha_{2},\ldots,\alpha_{n}]^{\mathrm{T}}.\label{eq:xk2Twithjordan}
\end{align}
Note that the eigenvalue $1$, which has algebraic multiplicity $\kappa\in\{1,2,\ldots,n\}$,
is a semisimple eigenvalue associated with the eigenvectors $v_{1},v_{2}\ldots,v_{\kappa}\in\mathbb{R}^{n}$.
As a result, $J_{i}=1$, $i\in\{1,2,\ldots,\kappa\}$, and hence $\lim_{k\to\infty}J_{i}^{k}=1$,
$i\in\{1,2,\ldots,\kappa\}$. On the other hand, for each $i\in\{\kappa+1,\kappa+2,\ldots,r\}$
the eigenvalue $\lambda_{i}$ is strictly inside the unit circle;
therefore, $\lim_{k\to\infty}J_{i}^{k}=0$, $i\in\{\kappa+1,\kappa+2,\ldots,r\}$.
Thus, $\lim_{k\to\infty}J^{k}$ takes the form of a diagonal matrix
with $1$ as the first $\kappa$ diagonal entries and $0$ elsewhere.
It follows from \eqref{eq:xk2Twithjordan} that 
\begin{align}
\lim_{k\to\infty}x(2kT) & =P\mathrm{diag}[\overbrace{1,\ldots,1}^{\kappa\,\mathrm{terms}},\overbrace{0,\ldots,0}^{n-\kappa\,\mathrm{terms}}]\nonumber \\
 & \quad\cdot[\alpha_{1},\alpha_{2},\ldots,\alpha_{n}]^{\mathrm{T}}\nonumber \\
 & =P[\alpha_{1},\ldots,\alpha_{\kappa},0,\ldots,0]^{\mathrm{T}}\nonumber \\
 & =\sum_{i=1}^{\kappa}\alpha_{i}v_{i},
\end{align}
 which completes the proof. \end{proof}

Under the condition in Theorem~\ref{theo1} that all the eigenvalues,
other than the semisimple eigenvalue $1$, of the monodromy matrix
$\Lambda$ are strictly inside the unit circle, the state evaluated
at integer multiples of the doubled period $2T$ converges to a point
on a periodic solution of the uncontrolled ($u(t)\equiv0$) system
\eqref{eq:system}. Hence, the state trajectory  converges towards
a periodic solution. The location of the limiting periodic solution
depends on the initial condition $x_{0}$. Specifically, the state
evaluated at integer multiples of the doubled period $2T$ converges
to the point given by $\sum_{i=1}^{\kappa}\alpha_{i}v_{i}$, where
$\alpha_{i}v_{i}$ characterizes the component of the initial state
$x_{0}$ along the eigenvector $v_{i}$ associated with the eigenvalue
$1$ of the monodromy matrix $\Lambda$. Note that if the algebraic
multiplicity of the eigenvalue $1$ is $\kappa=1$, then the limiting
periodic solution is given by $\alpha_{1}x^{*}(t)$, where $\alpha_{1}x^{*}(0)$
is the component of $x_{0}$ along the eigenvector $v_{1}=x^{*}(0)$. 

\section{Stabilization of Unstable Periodic Orbits of Nonlinear Systems }

\label{sec:Stabilization-of-Unstable}

We now employ the results obtained for linear periodic systems for
stabilizing an unstable periodic orbit of a nonlinear system. 

Consider the nonlinear system given by 
\begin{equation}
\dot{x}(t)=f(x(t))+u(t),\quad x(0)=x_{0},\quad t\geq0,\label{eq:nonl-system}
\end{equation}
where $x(t)\in\mathbb{R}^{n}$ is the state vector, $u(t)\in\mathbb{R}^{n}$
is the control input, and $f:\mathbb{R}^{n}\to\mathbb{R}^{n}$ is
a nonlinear function. Suppose that the uncontrolled ($u(t)\equiv0$)
system \eqref{eq:nonl-system} possesses a periodic solution $x(t)\equiv x^{*}(t)$
with a known period $T>0$ such that 
\begin{align}
x^{*}(t+T) & =x^{*}(t),\label{eq:upo}\\
\dot{x}^{*}(t) & =f(x^{*}(t)),\quad t\in[0,\infty).\label{eq:upoproperty2}
\end{align}
 The periodic orbit associated with the periodic solution $x^{*}(\cdot)$
is given by $\mathcal{O}\triangleq\{x^{*}(t)\colon t\in[0,T)\}$.
The stability of the periodic orbit $\mathcal{O}\subset\mathbb{R}^{n}$
is characterized through the stability of a fixed point of a Poincar\'e
map defined on an $(n-1)$-dimensional hypersurface that is transversal
to the periodic orbit (see \cite{guckenheimer2002,chicone2006}).
In this paper, we consider the case where $\mathcal{O}$ is an unstable
periodic orbit (UPO) and discuss its stabilization. 

There are several methods known for stabilizing the UPO. One of them
is the Pyragas-type delayed feedback control framework (see \cite{pyragas:ccc,tian2005survey,hovel2010control}).
In this framework the control input is given by 
\begin{equation}
u(t)=-F(x(t)-x(t-T)),\label{eq:pyragas-type-control}
\end{equation}
where $F\in\mathbb{R}^{n\times n}$ is the gain matrix of the controller.
The control input is computed based on the difference between the
current state and the delayed state. The delay time is set to correspond
to the period $T$ of the desired UPO so that the control input vanishes
after the UPO is stabilized. In the delayed feedback control method,
the controller uses the delayed state instead of the UPO as a reference
signal to which the current state is desired to be stabilized. Therefore
this method does not require a preliminary calculation of the UPO
if its period is given. 

The analysis of the closed-loop system under delayed feedback controller
is difficult, because the closed-loop dynamics is described by a delay-differential
equation \eqref{eq:nonl-system}, \eqref{eq:pyragas-type-control},
the state space of which is infinite-dimensional. This fact is our
motivation for employing the act-and-wait-fashioned delayed feedback
control law \eqref{eq:control-input}, since the closed-loop system
system \eqref{eq:control-input}, \eqref{eq:nonl-system} can be analyzed
by utilizing the methods that we developed in Section~\ref{sec:mainresults}. 

\begin{remark} Act-and-wait-fashioned delayed-feedback control laws
were previously used in \cite{konishi2011delayed} and \cite{pyragas2016}
for different problem settings. In \cite{konishi2011delayed}, the
stabilization of a fixed-point is considered. Furthermore, in \cite{pyragas2016},
periodic orbit stabilization is considered for a nonautonomous system.
Specifically, the uncontrolled system in \cite{pyragas2016} is affected
by an external periodic force, which induces the periodic orbit. The
controller in \cite{pyragas2016} is designed so that all Floquet
multipliers of the linearized system are strictly inside the unit
circle of the complex plane. In our case, the uncontrolled system
is not driven by a periodic force and it is autonomous. The periodic
orbit in our case is embedded in the dynamics. The stability assessment
method in this paper differs from that in \cite{pyragas2016} due
to the difference in the analysis of autonomous and nonautonomous
systems (see Section 7.1.3 of \cite{leine2013dynamics}). In particular,
as we discuss below, the linearized system in our case always possesses
$1$ as a Floquet multiplier regardless of the choice of the feedback
gain matrix, and moreover, the stability of the periodic orbit under
the act-and-wait-fashioned controller can be analyzed by assessing
the Floquet multipliers that are not $1$. \end{remark} 

We analyze the stability of the periodic orbit $\mathcal{O}\subset\mathbb{R}^{n}$
under the act-and-wait-fashioned control input by assessing the monodromy
matrix for the linear variational equation associated with the closed-loop
dynamics \eqref{eq:control-input}, \eqref{eq:nonl-system}. Specifically,
we linearize the closed-loop system \eqref{eq:control-input}, \eqref{eq:nonl-system}
around the periodic trajectory $x^{*}(t)$. First, we write the solution
of \eqref{eq:control-input}, \eqref{eq:nonl-system} as 
\begin{equation}
x(t)=x^{*}(t)+\delta x(t),\label{eq:variation}
\end{equation}
where $\delta x(t)$ is the state deviation from the periodic solution
$x^{*}(t)$ at time $t$. It then follows from \eqref{eq:nonl-system}
and \eqref{eq:variation} that 
\begin{align}
 & \dot{x}^{*}(t)+\dot{\delta x}(t)\nonumber \\
 & \quad=f(x^{*}(t)+\delta x(t))\nonumber \\
 & \quad\quad-g(t)F(x^{*}(t)+\delta x(t)-x^{*}(t-T)-\delta x(t-T))\nonumber \\
 & \quad=f(x^{*}(t))+{\left.\frac{\partial f(x)}{\partial x}\right|}_{x=x^{*}(t)}\delta x(t)+H.O.T.\nonumber \\
 & \quad\quad-g(t)F(\delta x(t)-\delta x(t-T)),\label{eq:hot}
\end{align}
where $H.O.T.$ denotes the higher-order terms in $\delta x$. 

By using \eqref{eq:upoproperty2} and neglecting the higher-order
terms for infinitely small deviations in \eqref{eq:hot}, we obtain
the linear variational equation
\begin{equation}
\dot{\delta x}(t)=A(t)\delta x(t)-g(t)F(\delta x(t)-\delta x(t-T)),\label{eq:variationalequation}
\end{equation}
where 
\begin{equation}
A(t)\triangleq{\left.\frac{\partial f(x)}{\partial x}\right|}_{x=x^{*}(t)}.\label{eq:defofat}
\end{equation}
Since $x^{*}(t)$ is a $T$-periodic trajectory, it follows from \eqref{eq:defofat}
that $A(t)$ is also $T$-periodic, that is, $A(t)=A(t+T)$. Hence,
the linear variational dynamics \eqref{eq:variationalequation} is
in fact characterized by the linear periodic system \eqref{eq:system}
(with $B(t)=I_{n}$) under the act-and-wait-fashioned delayed feedback
controller \eqref{eq:control-input}. Consequently, the monodromy
matrix $\Lambda$ associated with the linear variational dynamics
\eqref{eq:variationalequation} can be obtained by using \eqref{eq:monodromymatrix}. 

Next, we show that $\dot{x}^{*}(t)$ is a periodic solution of the
linear variational dynamics. First, by using \eqref{eq:nonl-system}
with $u(t)\equiv0$, we obtain 
\begin{align}
\frac{\mathrm{d}\dot{x}(t)}{\mathrm{d}t} & =\frac{\partial f(x(t))}{\partial x(t)}\frac{\mathrm{d}x(t)}{\mathrm{d}t}.\label{eq:derivatives}
\end{align}
Hence, for $x(t)\equiv x^{*}(t)$, it follows from \eqref{eq:defofat}
and \eqref{eq:derivatives} that 
\begin{align}
\frac{\mathrm{d}\dot{x}^{*}(t)}{\mathrm{d}t} & =A(t)\dot{x}^{*}(t).
\end{align}
Thus, the linear variational equation characterized by \eqref{eq:variationalequation}
has $\dot{x}^{*}(t)$ as a solution, that is, $\delta x(t)\equiv\dot{x}^{*}(t)$
satisfies \eqref{eq:variationalequation}. To see that $\dot{x}^{*}(t)$
is $T$-periodic, note that for $x(t)\equiv x^{*}(t)$, we have $u(t)\equiv0$,
and hence, \eqref{eq:upo} and \eqref{eq:upoproperty2} imply $\dot{x}^{*}(t+T)=f(x^{*}(t+T))=f(x^{*}(t))=\dot{x}^{*}(t)$,
$t\geq0$. 

Note that $\dot{x}^{*}(T)=\Phi(T,0)\dot{x}^{*}(0)$, where $\Phi(T,0)$
is the monodromy matrix for the uncontrolled system. Since, $\dot{x}^{*}(T)=\dot{x}^{*}(0)$,
it follows that $\dot{x}^{*}(0)$ is an eigenvector of the monodromy
matrix $\Phi(T,0)$ corresponding to the eigenvalue $1$. Note also
that
\begin{align}
\dot{x}^{*}(0) & =\dot{x}^{*}(2T)=\Lambda\dot{x}^{*}(0),\label{eq:nonlinearmonodromy}
\end{align}
where $\Lambda$ is the monodromy matrix associated with the closed-loop
linear variational dynamics \eqref{eq:variationalequation}. It follows
from \eqref{eq:nonlinearmonodromy} that $\Lambda$ possesses $1$
as an eigenvalue with the corresponding eigenspace $\{\lambda\dot{x}^{*}(0)\colon\lambda\in\mathbb{C}\}$
regardless of the choice of feedback gain matrix $F$ in the control
law \eqref{eq:control-input}. 

Note that the eigenvalues of the monodromy matrix $\Lambda$ that
are not associated with the eigenspace $\{\lambda\dot{x}^{*}(0)\colon\lambda\in\mathbb{C}\}$
correspond to the eigenvalues of a linearized Poincar\'e map, which
characterize the local asymptotic stability of the periodic orbit
$\mathcal{O}$ of the nonlinear system \eqref{eq:nonl-system} (see
\cite{guckenheimer2002}). 

\begin{remark} \label{RemarkForNonlinearStability}

Following the approach presented in \cite{chicone2006} and \cite{meirovitch2003},
we characterize the asymptotic stability of the periodic orbit $x^{*}(\cdot)$
of the closed-loop nonlinear system \eqref{eq:control-input}, \eqref{eq:nonl-system},
by assessing the spectrum of the monodromy matrix associated with
the linear variational dynamics \eqref{eq:variationalequation}. In
particular, the act-and-wait-fashioned Pyragas-type delayed feedback
control law \eqref{eq:control-input} asymptotically stabilizes the
periodic orbit $\mathcal{O}$ of \eqref{eq:nonl-system} if all the
eigenvalues, other than the eigenvalue $1$ associated with the eigenspace
$\{\lambda\dot{x}^{*}(0)\colon\lambda\in\mathbb{C}\}$, of the monodromy
matrix $\Lambda$ of the linear variational equation \eqref{eq:variationalequation}
are strictly inside the unit circle of the complex plane. \end{remark}

Note that the monodromy matrix of the linear variational equation
\eqref{eq:variationalequation} can be calculated using \eqref{eq:monodromymatrix}
with $B(t)=I_{n}$. For the periodic matrix $A(t)\in\mathbb{R}^{n\times n}$
and a feedback gain matrix $F\in\mathbb{R}^{m\times n}$, numerical
methods can be employed to calculate the eigenvalues of the monodromy
matrix $\Lambda$, which determine the asymptotic stability. Note
also that the conditions on the monodromy matrix of the linear variational
equation is only enough to guarantee local stability of the periodic
orbit of the nonlinear system \eqref{eq:control-input}, \eqref{eq:nonl-system}
(see Chapter 7 of \cite{leine2013dynamics}). For obtaining global
stability results, the higher-order terms in \eqref{eq:hot} have
to be taken into consideration.

\section{Illustrative Numerical Examples}

\label{sec:Illustrative-Numerical-Example}

In this section, we provide two numerical examples to demonstrate
our main results. 

\begin{figure}[t]
\begin{center}\psfrag{x1}[bc]{\hspace{0.2cm}\footnotesize{$x^{*}_{1}(t)$}}
\psfrag{x2}[bc]{\footnotesize{$x^{*}_{2}(t)$}}
\psfrag{tt}[c]{\footnotesize{$\mathrm{Time}\, [t]$}}\includegraphics[width=0.9\columnwidth]{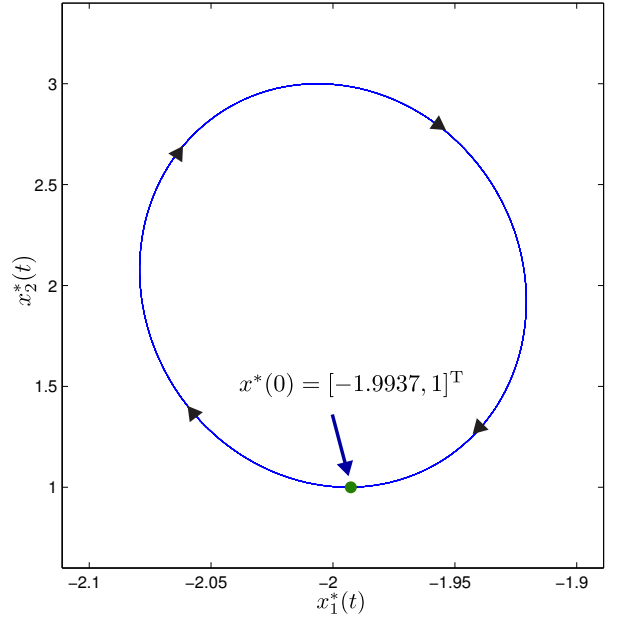}\end{center}\vskip -3pt\caption{Trajectory of the $T$-periodic solution $x^{*}(t)$ given by \eqref{eq:xstareq}}
\label{Flo:xstar}
\end{figure}

\begin{figure}[t]
\begin{center}\psfrag{x1}[bc]{\hspace{0.2cm}\footnotesize{$x_{1}(t)$}}
\psfrag{x2}[bc]{\footnotesize{$x_{2}(t)$}}
\psfrag{tt}[c]{\footnotesize{$\mathrm{Time}\, [t]$}}\includegraphics[width=0.9\columnwidth]{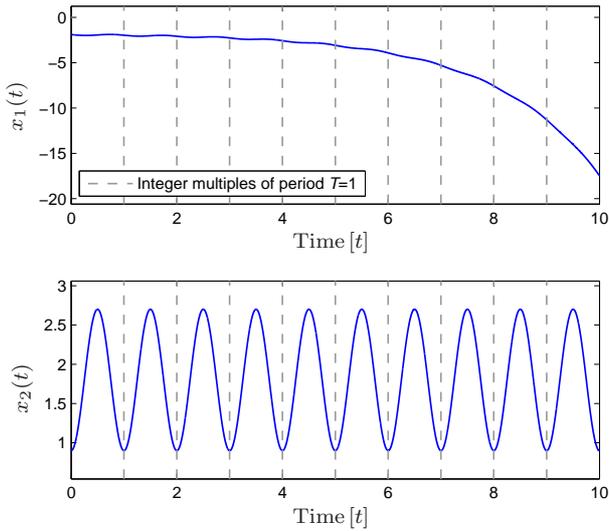}\end{center}\vskip -3pt\caption{State trajectories of the uncontrolled ($u(t)\equiv0$) system \eqref{eq:system}
for the initial condition $x_{0}=\left[-1.9,\,0.9\right]^{\mathrm{T}}$}
\label{Flo:x1-1}
\end{figure}

\begin{figure}
\begin{center}\includegraphics[width=0.9\columnwidth]{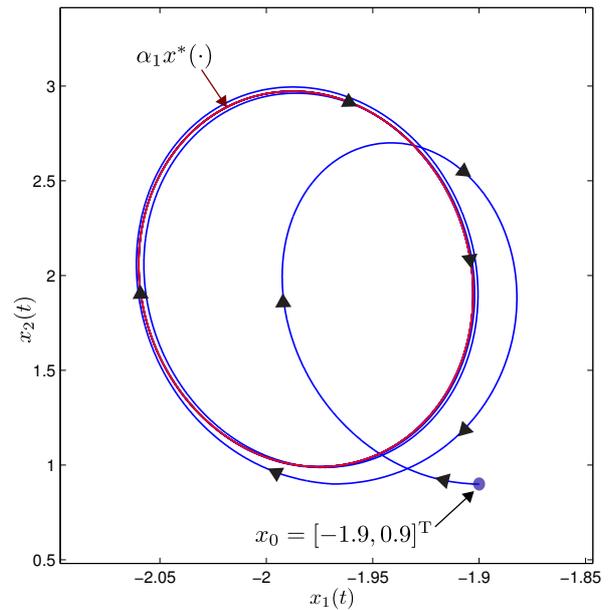}\end{center}\vskip -3pt\caption{Phase portrait of the closed-loop system \eqref{eq:system}, \eqref{eq:control-input}
obtained with the initial condition $x_{0}=\left[-1.9,\,0.9\right]^{\mathrm{T}}$ }
\label{Flo:ppclosedloop1}
\end{figure}

\begin{figure}
\begin{center}\psfrag{x1}[bc]{\hspace{0.2cm}\footnotesize{$x_{1}(t)$}}
\psfrag{x2}[bc]{\footnotesize{$x_{2}(t)$}}
\psfrag{x1ttt}{\scriptsize{$x_{1}(t)$}}
\psfrag{x2ttt}{\scriptsize{$x_{2}(t)$}}
\psfrag{a1xstar1ttt}{\scriptsize{$\alpha_{1}x^{*}_{1}(t)$}}
\psfrag{a1xstar2ttt}{\scriptsize{$\alpha_{1}x^{*}_{2}(t)$}}
\psfrag{tt}[c]{\footnotesize{$\mathrm{Time}\, [t]$}}\includegraphics[width=0.9\columnwidth]{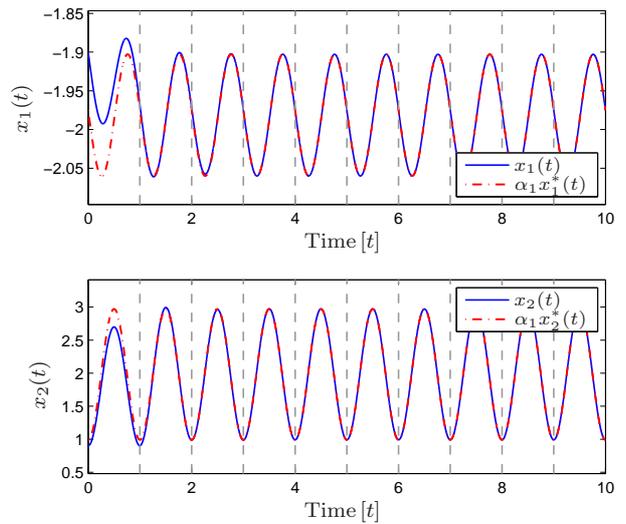}\end{center}\vskip -3pt\caption{State trajectories of the closed-loop system \eqref{eq:system}, \eqref{eq:control-input}
obtained with the initial condition $x_{0}=\left[-1.9,\,0.9\right]^{\mathrm{T}}$ }
\label{Flo:stclosedloop1}
\end{figure}

\begin{figure}
\begin{center}\psfrag{uu}[c]{\footnotesize{$u(t)$}}
\psfrag{tt}[c]{\footnotesize{$\mathrm{Time}\, [t]$}}\includegraphics[width=0.9\columnwidth]{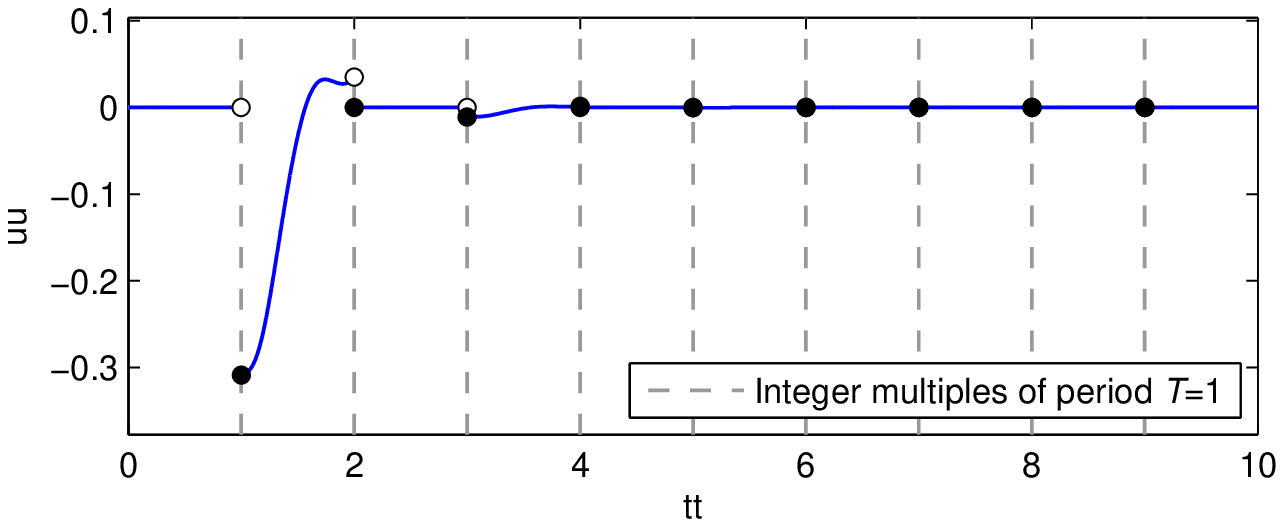}\end{center}\vskip -3pt\caption{Control input versus time }
\label{Flo:u1}
\end{figure}

\textbf{Example 4.1 } Consider the linear time-varying periodic
system \eqref{eq:system} described by periodic matrices 
\begin{align*}
A(t) & \triangleq\left[\begin{array}{cc}
0.5 & 0.5\\
0 & \frac{2\pi\sin(2\pi t)}{2-\cos(2\pi t)}
\end{array}\right],\,\,\,B(t)\triangleq\left[\begin{array}{c}
0\\
1+\sin^{2}(2\pi t)
\end{array}\right].
\end{align*}
The period of the time-varying system is $T=1$, that is, $A(t+1)=A(t)$
and $B(t+1)=B(t)$. The uncontrolled ($u(t)\equiv0$) dynamics possess
a $T$-periodic solution $x(t)\equiv x^{*}(t)$ given by 
\begin{align}
 & x^{*}(t)\triangleq\left[\begin{array}{c}
\frac{1}{(1+16\pi^{2})}(\cos(2\pi t)-4\pi\sin(2\pi t))-2\\
2-\cos(2\pi t)
\end{array}\right].\label{eq:xstareq}
\end{align}

Note that $x(t)\equiv\alpha x^{*}(t)$, where $\alpha\in\mathbb{R}$,
is also a $T$-periodic solution of the uncontrolled system, that
is, $x(t)\equiv\alpha x^{*}(t)$ satisfies \eqref{eq:system} with
$u(t)\equiv0$. Fig.~\ref{Flo:xstar} shows the trajectory of the
periodic solution $x(t)\equiv x^{*}(t)$. 

The monodromy matrix associated with the uncontrolled system is given
by 
\begin{align}
\Phi(T,0) & =\left[\begin{array}{cc}
1.6487 & 1.2934\\
0 & 1
\end{array}\right],
\end{align}
 which has the eigenvalues $1.6487$ and $1$. Note that the eigenvalue
$1.6487$ of the monodromy matrix $\Phi(T,0)$ lies outside the unit
circle. The periodic system without control input, hence, shows unstable
behavior (see Figure~\ref{Flo:x1-1} for state trajectories obtained
for the initial condition $x_{0}=\left[-1.9,\,0.9\right]^{\mathrm{T}}$). 

We are interested in finding a feedback gain matrix $F\in\mathbb{R}^{1\times2}$
such that the delayed-feedback control characterized in \eqref{eq:control-input}
guarantees convergence of the state trajectory towards a periodic
solution. In order to evaluate the asymptotic behavior of solutions
under the control law \eqref{eq:control-input} we need to examine
the eigenvalues of the monodromy matrix $\Lambda$ associated with
the closed-loop system. It is difficult to find an analytical expression
for the monodromy matrix $\Lambda$. For that reason, we numerically
calculate the value of $\Lambda$ for a certain range of feedback
gain parameters and search for a feedback gain matrix $F\in\mathbb{R}^{1\times2}$
that satisfies the condition in Theorem~\ref{theo1}. Note that for
the feedback gain matrix $F=[4.5,\,\,0.6]$, the corresponding monodromy
matrix is given by 
\begin{align}
\Lambda & =\left[\begin{array}{cc}
1.6857 & 1.3671\\
-0.8273 & -0.6495
\end{array}\right],\label{eq:phinewexample}
\end{align}
which has the eigenvalues $\lambda_{1}=1$ and $\lambda_{2}=0.0362$
associated with the eigenvectors $v_{1}=x^{*}(0)=[-1.9937,\,1]^{\mathrm{T}}$
and $v_{2}=[-0.6381,\,0.7699]^{\mathrm{T}}$, respectively. Note that
the eigenvalue $\lambda_{2}$ is inside the unit circle of the complex
plane. Therefore, it follows from Theorem~\ref{theo1} that under
the control law \eqref{eq:control-input}, the state trajectory evaluated
at integer multiples of the doubled period $2T$ converges to $\alpha_{1}v_{1}$,
where $\alpha_{1}v_{1}$ represents the component of a given initial
condition $x_{0}$ along the eigenvector $v_{1}=x^{*}(0)$. Figures~\ref{Flo:ppclosedloop1}
and \ref{Flo:stclosedloop1} respectively show the phase portrait
and state trajectories of the closed-loop system \eqref{eq:system},
\eqref{eq:control-input} for the initial condition $x_{0}=[-1.9,\,0.9]^{\mathrm{T}}$,
which can be represented as $x_{0}=\alpha_{1}v_{1}+\alpha_{2}v_{2}$
where $\alpha_{1}=0.9907$ and $\alpha_{2}=-0.1178$. Note  that $\lim_{k\to\infty}x(2kT)=\alpha_{1}v_{1}=\alpha_{1}x^{*}(0)$
and hence the state trajectory converges to the periodic solution
$\alpha_{1}x^{*}(t)$. Note that the convergence is achieved with
the help of the proposed delayed feedback controller, which is turned
on and off alternately at every integer multiples of the period $T=1$
of the uncontrolled system, and hence the control input (shown in
Figure~\ref{Flo:u1}) is discontinuous at time instants $T$, $2T$,
$3T$, $\ldots$. Note also that the control input converges to $0$
as the state converges to the periodic solution. 
\begin{figure}
\begin{center}\includegraphics[width=0.9\columnwidth]{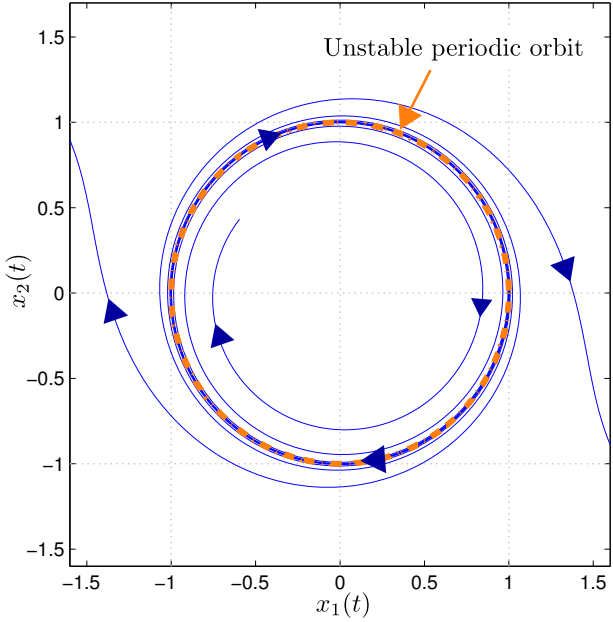}\end{center}
\vskip -10pt\caption{Phase portrait of the uncontrolled ($u(t)\equiv0$) nonlinear system
\eqref{eq:nonl-system} }
\label{unstablephaseportrait}
\end{figure}

\textbf{Example 4.2 } In this example we demonstrate the utility
of our proposed control framework for the stabilization of an unstable
periodic orbit of a nonlinear system. Specifically, we consider the
nonlinear dynamical system \eqref{eq:nonl-system} with 
\begin{align}
f(x) & \triangleq\left[\begin{array}{c}
-x_{1}\big(\varphi(x)-\varphi^{2}(x)\big)+2\pi x_{2}\\
-x_{2}\big(\varphi(x)-\varphi^{2}(x)\big)-2\pi x_{1}
\end{array}\right],\label{eq:examplef}
\end{align}
 where $\varphi(x)\triangleq x_{1}^{2}+x_{2}^{2}$. This system is
a modified version of an example nonlinear dynamical system considered
in Section 2.7 of \cite{khalilbook}. The phase portrait of the uncontrolled
($u(t)\equiv0$) nonlinear system \eqref{eq:nonl-system} is shown
in Figure \ref{unstablephaseportrait}. The system has clockwise-revolving
unstable periodic orbit $\mathcal{O}\triangleq\{x^{*}(t)\colon t\in[0,T)\}$,
where $x^{*}(t)=[\cos2\pi t,\,-\sin2\pi t]^{\mathrm{T}}$ is a $1$-periodic
solution of the uncontrolled system. The linear variational equation
associated with the closed-loop system \eqref{eq:nonl-system} under
our proposed controller \eqref{eq:control-input} (with $T=1$) is
given by \eqref{eq:variationalequation}, where 
\begin{align}
A(t) & =\left[\begin{array}{cc}
a_{1,1}(t) & a_{1,2}(t)\\
a_{2,1}(t) & a_{2,2}(t)
\end{array}\right]
\end{align}
 with 
\begin{align}
a_{1,1}(t) & =2\cos^{2}(2\pi t),\\
a_{1,2}(t) & =-2\cos(2\pi t)\sin(2\pi t)+2\pi,\\
a_{2,1}(t) & =-2\cos(2\pi t)\sin(2\pi t)-2\pi,\\
a_{2,2}(t) & =2\sin^{2}(2\pi t).
\end{align}
Note that it is difficult to find an analytical expression for the
monodromy matrix $\Lambda$ associated with the linear variational
equation \eqref{eq:variationalequation}. For this reason, we numerically
calculate the value of $\Lambda$ for a certain range of the elements
of the gain matrix. In particular, for the case 
\begin{align}
F & =\left[\begin{array}{cc}
0.7 & \,\,4.1\\
-4.1 & \,\,0.7
\end{array}\right],\label{eq:exampleF}
\end{align}
 the corresponding monodromy matrix is given by 
\begin{align}
\Lambda & =\left[\begin{array}{cc}
-0.0093 & \,\,0\\
-6.5898 & \,\,1
\end{array}\right],
\end{align}
 which has the eigenvalues $\lambda_{1}=-0.0093$, $\lambda_{2}=1$,
and eigenvectors $v_{1}=[-0.1514,\,-0.9885]^{\mathrm{T}}$, $v_{2}=[0,\,-1]^{\mathrm{T}}$.
Note that the eigenvalue $\lambda_{1}$ (eigenvalue that is not $1$)
is inside the unit circle. Therefore, the feedback control law \eqref{eq:control-input}
with the feedback gain matrix \eqref{eq:exampleF} asymptotically
stabilizes the periodic orbit $x^{*}(\cdot)$ of the nonlinear system
\eqref{eq:nonl-system} (see Remark~\ref{RemarkForNonlinearStability}). 

\begin{figure}
\begin{center}\includegraphics[width=0.9\columnwidth]{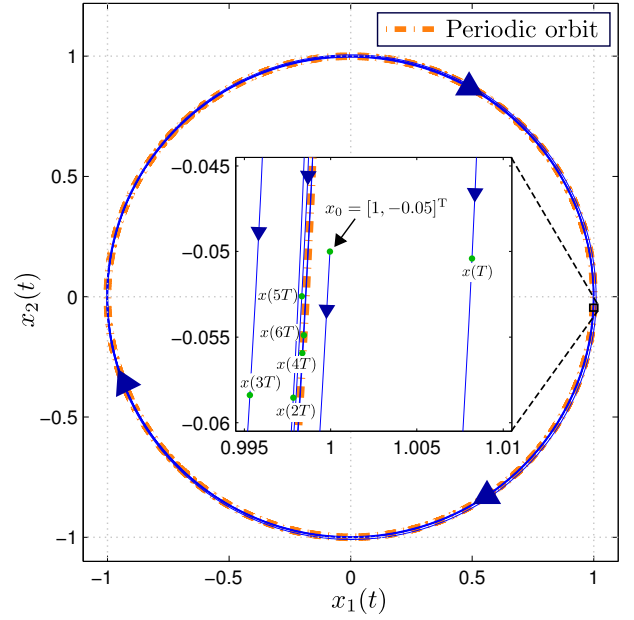}\end{center}
\vskip -10pt\caption{Phase portrait of the closed-loop system \eqref{eq:control-input},
\eqref{eq:nonl-system}. Inset figure shows the initial condition
$x_{0}=[1,-0.05]^{\mathrm{T}}$ as well as the state evaluated at
integer multiples of $T$. }
\label{stablephaseportrait}
\end{figure}

Figures~\ref{stablephaseportrait} and \ref{statemagexample} respectively
show the phase portrait and the state magnitude of the closed-loop
system \eqref{eq:control-input}, \eqref{eq:nonl-system} obtained
for the initial condition $x_{0}=[1,-0.05]^{\mathrm{T}}$. The state
trajectory converges to the periodic orbit and the state magnitude
$\|x(t)\|$ converges to the desired value $1$. 

Note that in this paper, we investigate local stability of the periodic
orbit of the nonlinear system through a linearization approach. The
initial condition $x_{0}$ in the simulation is selected close to
$x^{*}(0)$, so that $\delta x(t)$ (the state deviation from the
periodic solution) remains small and the effect of higher-order terms
in the variational equation \eqref{eq:hot} is negligible. Note that,
if the initial state is far from the orbit, then the control law \eqref{eq:control-input}
with the feedback gain $F$ may no longer achieve stabilization, since
the higher order terms in the variational equation may have strong
effects on the state trajectory. For achieving global stabilization,
the higher-order terms in \eqref{eq:hot} have to be taken into consideration
for control design. 

\begin{figure}
\begin{center}\psfrag{magx}[bc]{\footnotesize{$\| x(t) \|$}}
\psfrag{tt}[c]{\footnotesize{$\mathrm{Time}\, [t]$}}\includegraphics[width=0.9\columnwidth]{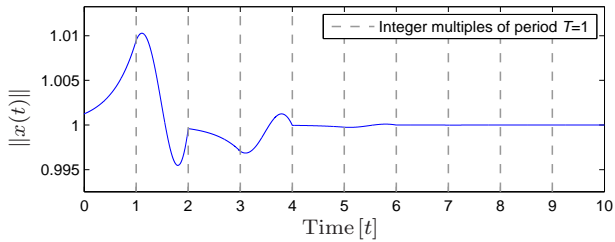}\end{center} \vskip -10pt\caption{State magnitude versus time }
\label{statemagexample}
\end{figure}

\begin{figure}
\begin{center}\psfrag{uu1}[c]{\footnotesize{$u_{1}(t)$}}
\psfrag{uu2}[c]{\footnotesize{$u_{2}(t)$}}
\psfrag{tt}[c]{\footnotesize{$\mathrm{Time}\, [t]$}}\includegraphics[width=0.9\columnwidth]{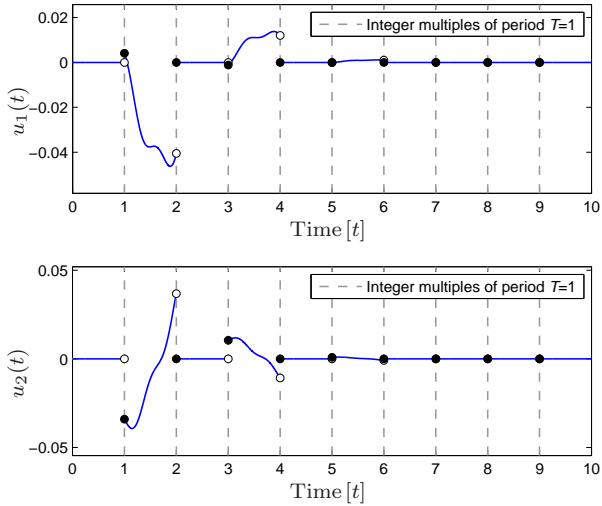}\end{center}\vskip -10pt\caption{Control input versus time }
\label{Flo:u2}
\end{figure}

The control input trajectory (shown in Figure~\ref{Flo:u2}) is discontinuous
at time instants $T,2T,3T,\ldots$. At these time instants, the delayed-feedback
controller is turned on and off alternately according to the switching
function $g(t)$ defined in \eqref{eq:gdef}. For the $2T$-periodic
switching function $g(t)$, the closed-loop system is $2T$-periodic;
hence, the stabilization of the UPO could be characterized through
the monodromy matrix $\Lambda$ of the $2T$-periodic closed-loop
linear variational equation \eqref{eq:variationalequation}. 

Note that stabilization of the UPO could also be achieved through
utilization of alternative switching sequences that are different
from the one induced by $g(t)$. For example, one can consider the
switching function 
\begin{align}
\hat{g}(t) & \triangleq\left\{ \begin{array}{rl}
0, & 3kT\leq t<(3k+1)T,\\
1, & (3k+1)T\leq t<3(k+1)T,
\end{array}\right.\,k\in\mathbb{N}_{0}.\label{eq:gtildedef}
\end{align}
In this case the act-and-wait fashioned control law is given by 
\begin{align}
u(t) & \triangleq-\hat{g}(t)F(x(t)-x(t-T)).\label{eq:uwithghat}
\end{align}
Note that the closed-loop system under the controller \eqref{eq:uwithghat}
is $3T$-periodic, thus by analyzing the spectrum of the monodromy
matrix associated with the $3T$-periodic closed-loop linear variational
equation, we can assess whether the control law \eqref{eq:uwithghat}
guarantees local asymptotic stabilization of the UPO or not. 

According to the switching sequence induced by $\hat{g}(t)$, the
controller is active two-thirds of the time in average. On the other
hand, in the case of $g(t)$ defined in \eqref{eq:gdef}, the controller
is active only half of the time. Hence, one may think that a feedback
gain matrix that guarantees stabilization of the UPO for the case
of $g(t)$ guarantees stabilization also for the case of $\hat{g}(t)$.
However, this is not true. A feedback gain that guarantees stabilization
for one switching sequence does not necessarily guarantee stabilization
for another. For instance, the control input \eqref{eq:control-input}
with the gain matrix $F$ given by \eqref{eq:exampleF} achieves stabilization
(as illustrated in Figures \ref{stablephaseportrait} and \ref{statemagexample}),
whereas the control input \eqref{eq:uwithghat} with the same gain
matrix $F$ does not stabilize the UPO. In fact, the monodromy matrix
of the $3T$-periodic closed-loop linear variational equation under
the control law \eqref{eq:uwithghat} possesses the eigenvalue $-25.9876$,
which is outside the unit circle. 

In addition to utilizing a different switching sequence for the act-and-wait
controller, we may also employ a different delay term for the feedback
for stabilizing the UPO. For example, consider the act-and-wait fashioned
control law 
\begin{align}
u(t) & \triangleq-\bar{g}(t)F(x(t)-x(t-2T)),\label{eq:ubar}
\end{align}
 with the switching function 
\begin{align}
\bar{g}(t) & \triangleq\left\{ \begin{array}{rl}
0, & 4kT\leq t<(4k+2)T,\\
1, & (4k+2)T\leq t<4(k+1)T,
\end{array}\right.\,k\in\mathbb{N}_{0}.\label{eq:gbardef}
\end{align}
 Note that in this case the control law \eqref{eq:ubar} harnesses
the difference between the state at time $t$ (current state) and
the state at time $t-2T$. Furthermore, with the switching function
$\bar{g}(t)$, the controller is turned on and off alternately at
time instants $2T,4T,6T,\ldots$. The closed-loop system under the
control law \eqref{eq:ubar} is $4T$-periodic. It is important to
note here that the control input \eqref{eq:ubar} with a feedback
gain $F$ does not necessarily achieve stabilization of the UPO, even
if the control law \eqref{eq:control-input} with the same feedback
gain guarantees stabilization. For example, with the feedback gain
\eqref{eq:exampleF}, the control law \eqref{eq:control-input} achieves
stabilization of the UPO, whereas the control law \eqref{eq:ubar}
does not. In fact, the monodromy matrix of the linear variational
equation associated with the $4T$-periodic closed-loop system under
the control law \eqref{eq:ubar} with the feedback gain \eqref{eq:exampleF}
has the eigenvalue $69.4489$, which is outside the unit circle. The
discussion above illustrates that it is possible to obtain different
control laws by changing the act-and-wait sequence and/or changing
the delayed-feedback term; furthermore, each case with a different
control law requires independent analysis for assessing asymptotic
stabilization of the UPO. 

In practical implementations of our act-and-wait-fashioned control
laws, the timing of the switching may not always be exact. Furthermore,
it may also be the case that the delay time and the period of the
orbit do not exactly match. The effects of such practical issues require
further analysis. We note that for the standard delayed feedback control,
the effect of mismatches between the delay time and the period of
the orbit was analyzed in \cite{purewal2014effect}. We also note
that although the trajectories of the act-and-wait-fashioned control
law has discontinuities, this does not cause a practical problem in
the form of chattering, because the switching in the control law happens
only periodically with a period that is an integer multiple of the
period of the orbit to be stabilized. 

\section{Conclusion}

\label{sec:Conclusion}

We explored stabilization of the periodic orbits of linear periodic
time-varying systems through an act-and-wait-fashioned delayed feedback
control framework. Our proposed framework employs a switching mechanism
to turn the control input on and off alternately at every integer
multiples of the period of the desired orbit. The use of this mechanism
allows us to derive the monodromy matrix associated with the closed-loop
dynamics. By analyzing the eigenvalues of the monodromy matrix, we
obtained conditions under which the state trajectory converges towards
a periodic solution. We then applied our results in stabilization
of unstable periodic orbits of nonlinear systems. We discussed alternative
switching sequences for turning the controller on and off in our control
framework. We observe that a feedback gain that guarantees stabilization
for one switching sequence does not necessarily guarantee stabilization
for another. It is also interesting to observe that increasing the
amount of time the control input is turned on does not necessarily
increase the performance; in fact it may result in instability. 

In the act-and-wait-fashioned control laws that we considered, the
waiting duration where the control input is turned off is larger than
or equal to the delay amount. It is shown in \cite{insperger2011}
that the act-and-wait approach can also be useful in obtaining finite-dimensional
monodromy matrices even if the waiting duration is smaller than the
delay. One of our future research directions is to investigate the
utility of small waiting and acting durations in the delayed feedback
control of periodic orbits. 

\balance

\section*{Acknowledgements}

This research was supported by the Aihara Project, the FIRST program
from JSPS, initiated by CSTP.

\bibliographystyle{elsarticle-num}
\bibliography{references}

\end{document}